\documentclass[aps,onecolumn,superscriptaddress,dvipsnames,amssymb,9pt,floatfix]{revtex4-2}

\usepackage{amssymb}
\usepackage{amsmath}
\usepackage{array}
\usepackage{geometry}
\usepackage{amsthm, hyperref}
\usepackage{latexsym,bbm}
\usepackage{algorithm,algpseudocode}
\usepackage{enumerate}
\usepackage{float}
\usepackage{multirow}
\usepackage{fullpage}
\usepackage{makecell}
\usepackage{tikz}

\newcommand{\Ord}[1]{\mathcal{O}\left( #1 \right)}

\newtheorem{lemma}{Lemma}
\newtheorem{theorem}{Theorem}
\newtheorem{corollary}{Corollary}

\newtheorem{defn}{Definition}
\newtheorem{prob}{Problem}
\newtheorem{prob*}{Problem*}

\def\be{\begin{eqnarray}}
\def\ee{\end{eqnarray}}

\usepackage{color}
\definecolor{Pr}{rgb}{0.4,0.3,0.9}

\date{\today}

\begin{abstract}
Measuring properties of quantum systems is a fundamental problem in quantum mechanics.
We provide a simple method for estimating the expectation value of observables with an unknown quantum state.
The idea is to use a data structure to sample the terms of observables based on the Pauli decomposition proportionally to their importance.
We call this technique operator shadow as a shorthand for the procedure of preparing a sketch of an operator to estimate properties. 
Only when the numbers of observables are small for multiple local observables, the sample complexity of this method is better than the classical shadow technique.
However, if 
we want to estimate the expectation value of a linear combination of local observables, for example the energy of a local Hamiltonian, the sample complexity is better on all parameters. The time complexity to construct the data structure is $2^{O(k)}$ for $k$-local observables, similar to the post-processing time of classical shadows. 
\end{abstract}

\begin{document}
\title{Estimating properties of a quantum state by importance-sampled operator shadows}

\author{Naixu Guo}
\email{naixug@u.nus.edu}
\affiliation{Centre for Quantum Technologies, National University of Singapore, 117543, Singapore}

\author{Feng Pan}
\email{feng\_pan@nus.edu.sg}
\affiliation{Centre for Quantum Technologies, National University of Singapore, 117543, Singapore}

\author{Patrick Rebentrost}
\email{cqtfpr@nus.edu.sg}
\affiliation{Centre for Quantum Technologies, National University of Singapore, 117543, Singapore}
\affiliation{Department of Computer Science, School of Computing, National University of Singapore, 117417, Singapore}
\maketitle
\section{Introduction}

Closed quantum systems are essentially black boxes.
We may only extract classical information via measurements, which leads to the collapse of quantum states.
To recover a full description of an unknown $n$-qubit quantum state, in general exponentially many copies of the quantum state are needed \cite{BRU1999249}.
Standard techniques of quantum state tomography are reviewed in \cite{dariano2003quantum, ZhuThesis2012, wrightHowLearnQuantum}. In the work \cite{Flammia_2012}, authors considered compressed sensing which extracts a sketch of a quantum state from a few randomized measurements. 
O’Donnell and Wright \cite{10.1145/2897518.2897544}, Haah \textit{et al.} \cite{10.1145/2897518.2897585} respectively improved algorithms for the trace distance and fidelity settings.
Generalizing previous quantum-state PAC-learning \cite{Aaronson_2007}, the online learning setting has also been considered \cite{Aaronson_2019, chen2022adaptive, chen2020practical}.

For many practical applications, we may not need the full description of a quantum state. 
Aaronson considered a different setting called \textit{shadow tomography} \cite{10.1145/3188745.3188802}.
Given $M$ two-outcome measurements $\{E_1,\dots, E_m\}$ and many copies of an unknown quantum state $\rho$, it suffices to use only $\mathcal{O}(\log M)$ copies estimating each $\mathrm{Tr}[E_i\rho]$ with bounded precision.
Later, Huang \textit{et al.} generalized this problem to predict expectation values of multiple observables and had a great impact both theoretically and experimentally \cite{huang_predicting_2020, Huang2022exp}.
The key idea is to construct a classical sketch of the quantum state from the outcomes of random measurements.
They provided a detailed discussion for random Clifford and Pauli measurement settings.
Recently, Bertoni \textit{et al.} generalized to the random shallow circuit setting \cite{bertoni2023shallow}. 
Some physically oriented problems like the ground state prediction and learning phases of matter have also been explored \cite{Huang2022phaselearning, lewis2023improved}.

In this paper, we provide a straightforward method for estimating expectation values of observables based on Pauli measurements. 
We call the method \textit{importance-sampled operator shadows}. Analogous to the shadow of a quantum state cast by (randomized)  measurements, we consider the shadow of an operator cast by importance sampling of its constituents.
The method samples the Pauli measurements according to the size of the given coefficients in an $\ell_1$-sampling sense and estimates the desired expectation values.
Similar strategies have been explored in settings like Hamiltonian simulation, variational quantum algorithms, and direct fidelity estimation \cite{, PhysRevLett.129.030503, arrasmith2020operator, PhysRevLett.123.070503, PhysRevLett.106.230501}.
Especially, we improve on earlier work on direct fidelity estimation with $\ell_2$ sampling in Ref. \cite{PhysRevLett.106.230501}.
We provide a comprehensive analysis of our method, both for sample and time complexity, and make comparisons with the classical shadow method \cite{huang_predicting_2020}.
Numerical simulations are shown to corroborate our theoretical findings. For multiple local observables, the sample complexity of this method scales better than the classical shadow in other parameters (such as norms) yet exponentially worse on the number of observables.
However, if we want to estimate expectation values for the linear combination of local observables, the sample complexity of this method is better on all parameters. 

\section{Preliminaries}
For a matrix $A\in \mathbb{C}^{N\times N}$, let $a_{ij}$ be the $(i,j)$-element of $A$.
In this work, $\|A\|_F=\sqrt{\sum_i\sum_j |a_{ij}|^2}$ is the Frobenius norm (Hilbert-Schmidt norm).
Also, $\|A\|_{\infty}=\max_i \sigma_i(A)$ is the spectral norm, where $\sigma_i(A)$ are singular values of $A$.
For a vector $h\in \mathbb{R}^N$, we use $\|h\|_p$ to represent its $\ell_p$ norm.
A standard norm identity is
\be
\|h\|_2\leq \|h\|_1\leq \sqrt{N}\|h\|_2 \label{eqL1}.
\ee
Another simple norm inequality is the following, and the proof is described in Appendix \ref{appen.norm}.
\begin{lemma} \label{lemmaNorm}
For Hermitian matrix $A \in \mathbb C^{N\times N}$, we have that $\|A\|_F\leq \sqrt{N} \|A\|_{\infty}$.
\end{lemma}

\section{Importance-sampled operator shadow}

For a given $2^n$-dimensional observable $O$, we consider its Pauli decomposition, i.e.,
\be
O = \sum_{i=1}^{L} h_i P_i, \label{eqPauli}
\ee
where $P_i$ is a Pauli string over $n$ qubits, and $h_i$ are corresponding coefficients. 
Since observables are all Hermitians, we have $h_i\in \mathbb{R}$ for all $i\in [L]$.
In general, the decomposition is not efficient, i.e., $L = 4^n$.
In terms of the locality of the operator, we use the definition of locality that was also used in, e.g., Ref.~\cite{huang_predicting_2020}. 
Operator $O$ is called $k$-local if it acts nontrivially on at most $k$ qubits. Then, a $k$-local observable is decomposed into at most $L=4^k$ Pauli terms.
Under this definition, for example, the Hamiltonian for an $n$-spin Ising model with a complete graph would be $n$-local.
In addition, the Hamiltonian for an $n$-spin Ising model with 2D nearest neighbor couplings would also be $n$-local. However, by terminology used elsewhere a $k$-local Hamiltonian means a Hamiltonian consisting of terms that each act on $k$ qubits.

Assume we are given multiple copies of an unknown quantum state $\rho$.
We would like to estimate the expectation value by estimating expectation values of Pauli operators instead, i.e., 
\be \mathrm{Tr}[O\rho]
=\sum_{i=1}^L h_i\mathrm{Tr}[P_i \rho].
\ee 

A straightforward method to measure the expectation value is as follows.
Provide an estimate $z_i$ for each $\mathrm{Tr}[P_i \rho]$ to additive accuracy $\epsilon/\Vert h\Vert_1$ using a number of copies of $\Ord{\Vert h\Vert_1^2/\epsilon^2}$, see Lemma \ref{lemmaPauliMeasurement} below. Hence, this estimates $\mathrm{Tr}[O\rho]$ with additive error $\epsilon$, 
as seen from $\vert \mathrm{Tr}[O\rho]-\sum_{i=1}^L h_i z_i\vert \leq \sum_{i=1}^L \vert h_i\vert\ \vert \mathrm{Tr}[P_i \rho] - z_i \vert \leq \epsilon$. In total, we require $\Ord{L \Vert h\Vert_1^2/\epsilon^2}$ copies.
We show the improvement due to importance sampling.  
We first define sampling access inspired by Ref.~\cite{10.1145/1039488.1039494} where it was used in a singular value decomposition context. 
Note that dual access model (utilizing query and sample accesses) has been investigated in classical distribution testing like Ref.~\cite{canonne2014testing}.

\begin{defn}[$\ell_1$-sampling access for $O$]\label{defSamplingL1}
Define the sampling access such that we are able to sample an index $i \in [L]$ with probability 
\be
W_i :=\frac{ \vert h_i\vert  }{\Vert h \Vert_1 }.
\ee
\end{defn} 
For the expectation value estimation problem, we start with a perfect query model for each of the $\mathrm{Tr}[P_i \rho]$. 
The lemma is essentially a corollary of the results of \cite{10.1145/1039488.1039494,10.1145/3313276.3316310, gilyen2018quantuminspired, chia2018quantuminspired}.

\begin{lemma}[Perfect $\ell_1$-sampling] \label{lemmaSamplingL1}
Assume sampling access as in Definition \ref{defSamplingL1}. Also assume query access to $h_i$ and  $\mathrm{Tr}[P_i\rho]$ for all $i \in [L]$. In addition, assume knowledge of the norm $ \Vert h \Vert_1$. Let $\epsilon>0$.
Then we can estimate 
$\mathrm{Tr}[O\rho]$ to additive accuracy $  \|h\|_1 \epsilon$ in $\Ord{1/\epsilon^2}$ samples and queries with constant success probability. 
\end{lemma}
\begin{proof}
The algorithm is as follows.
We form the random variable 
\begin{align}
X:= \mathrm{sign}(h_i)\mathrm{Tr}[P_i\rho]\|h\|_1,
\end{align}
where the index $i \in [L]$ is sampled according to $W_i$, and query the respective elements.
Considering the randomness of $i$, the expectation value is 
\begin{align}
\mathbb{E}[X]=\sum_{i=1}^L W_i\mathrm{sign}(h_i)\mathrm{Tr}[P_i\rho]\|h\|_1=\mathrm{Tr}[O\rho].
\end{align}
The variance is bounded as 
\be
\mathbb{V}[X]&\leq& \sum_{i=1}^L W_i \left(\mathrm{sign}(h_i)\mathrm{Tr}[P_i\rho] \|h\|_1\right)^2
=\sum_{i=1}^L \frac{|h_i|}{\|h\|_1}(\mathrm{Tr}[P_i\rho] \|h\|_1)^2
= \|h\|_1 \sum_{i=1}^L |h_i| (\mathrm{Tr}[P_i\rho])^2 \nonumber\\
&\leq& \|h\|_1^2.
\ee
Consider performing the sampling independently $T$ times. 
Let $S:=\frac{1}{t}\sum_{t=1}^T X_t$.
By Chebyshev's inequality, 
\begin{align}
\mathbb{P}[|S-\mathbb{E}[S]|\geq \|h\|_1 \epsilon ]\leq \frac{1}{\epsilon^2 T}.
\end{align}
To obtain a constant failure probability (say $0.01$), it suffices to take $T\in\mathcal{O}(100/\epsilon^2)$.
The dependence on the failure probability can be exponentially improved via median boosting \cite{10.5555/646239.683352}.
\end{proof}

Similar results can also be found in Ref.~\cite{PhysRevLett.127.200501, Wu2023overlappedgrouping}.
In the physical situation, the perfect query model is an unrealistic assumption. Multiple runs of the experiment measuring a certain operator $P_i$ return the expectation value only with certain accuracy. We can nevertheless hope to obtain an efficient estimate of the expectation value 
$\mathrm{Tr}[O\rho]$.
The error of each Pauli measurement is given by the following straightforward lemma.
\begin{lemma} \label{lemmaPauliMeasurement}
Let $\epsilon>0$ and $P$ be any Pauli operator over $n$ qubits. Let multiple copies of an arbitrary $n$-qubit quantum state $\rho$ be given. The expectation value $\mathrm{Tr}[P\rho]$ can be determined to additive accuracy $\epsilon$ using $\Ord{1/\epsilon^2}$ copies of $\rho$. 
\end{lemma}
\begin{proof}
A single measurement of $P$ obtains the outcome $\pm 1$. We have $\mathrm{Tr}[P\rho] = 2 p -1$, where $p$ is the probability of measuring $+1$. To estimate this probability we perform a simple Bernoulli test. We define a random variable $Z^{(j)}$ for the $j$-th measurement, where $Z^{(j)}=1$ if the outcome is positive and $Z^{(j)}=0$ if the outcome is negative.
Suppose we measure $M$ times independently.
Let $Z=\sum_{j=1}^M Z^{(j)}/M$.
This statistic has expectation value $\mathbb{E}[Z] = p$ and variance $\mathbb{V}[ Z] = p(1-p)/M$. 
By Chebyshev's inequality, we have 
$\mathbb{P}\left [ Z - \mathbb{E}[Z] \vert \geq \epsilon\right] \leq \frac{ \mathbb{V}[Z] }{\epsilon^2} = \frac{ p(1-p) }{M \epsilon^2}$. For a constant failure probability $\delta$, we require at least $ M \geq p(1-p)/(\delta \epsilon^2) =(1 - \vert \mathrm{Tr}[P\rho] \vert^2)/(\delta \epsilon^2)$.
Assuming the worst case $\vert \mathrm{Tr}[P\rho] \vert =0$ and since there are no other constraints on $M$, we have that $ M \in \Ord{1/\epsilon^2}$.
\end{proof}
Combining these lemmas leads to the following result. 

\begin{theorem}[$\ell_1$-sampling operator shadow]\label{theoremL1}
Assume sampling access as in Definition \ref{defSamplingL1}.
Also, assume query
access to $h_i$ for all $i\in [L]$ and knowledge of the norm $ \Vert h \Vert_1$.
Let there be given multiple copies of a quantum state $\rho$. Let $\epsilon>0$.
Then there is an algorithm that can estimate
$\mathrm{Tr}[O\rho]$ to additive accuracy $\epsilon$ with $ \Ord{\|h\|_1^2/\epsilon^2}$ measurements with constant success probability.
\end{theorem}

\begin{proof}
The algorithm is given in Algorithm \ref{algorithmL1}. Evaluate the correctness and run time as follows. Define the random variable 
\begin{align}
\tilde{X}:=\mathrm{sign}(h_i)\left(2Z_i-1\right)\|h\|_1,
\end{align}
where the index $i \in [L]$ is sampled according to $W_i$.
The definition of $Z_i$ is the statistic $Z$ of Lemma \ref{lemmaPauliMeasurement} for operator $P_i$. 
Considering the randomness of $i$ and the randomness of $Z_i$, we have 
\begin{align}
\mathbb{E}[\tilde{X}]=\sum_{i=1}^L W_i \mathrm{sign}(h_i)(2\mathbb{E}[Z_i]-1) \|h\|_1= \sum_{i=1}^L W_i \mathrm{sign}(h_i)\mathrm{Tr}[P_i\rho] \|h\|_1=\mathrm{Tr}[O\rho].
\end{align}
Note that 
\be
\mathbb{E}[Z_i^2]&=&\mathbb{V}[Z_i]+\mathbb{E}[Z_i]^2=\frac{p_i}{M}+\left(1-\frac{1}{M}\right)p_i^2\nonumber\\
&=&\frac{1+\mathrm{Tr}[P_i\rho]}{2M}+\left(1-\frac{1}{M}\right)\left(\frac{1+\mathrm{Tr}[P_i\rho]}{2}\right)^2.
\ee
The total variance can be bounded as (using $(a-b)^2\leq a^2+b^2$ for $a,b\geq 0$)
\be
\mathbb{V}[\tilde{X}]\leq \|h\|_1^2\left(4\mathbb{E}\left[Z_i^2\right]+1\right)
&=&\|h\|_1^2 \left(4\sum_{i=1}^L W_i \left(\frac{1+\mathrm{Tr}[P_i\rho]}{2M}+\left(1-\frac{1}{M}\right)\left(\frac{1+\mathrm{Tr}[P_i\rho]}{2}\right)^2\right)+1\right)\nonumber \\
&\leq& \|h\|_1^2  \left(4\sum_{i=1}^L W_i \frac{1+\mathrm{Tr}[P_i\rho]}{2M}+5\right)  \leq \|h\|_1^2 \left(\frac{4}{M}+5\right).
\ee  
Consider performing this $T$ times independently. Let $\tilde{S}:=\frac{1}{t}\sum_{t} \tilde{X}_t$.
By Chebyshev's inequality, 
\begin{align}
\mathbb{P}[|\tilde{S}-\mathbb{E}[\tilde{S}]|\geq \epsilon]\leq \frac{\mathbb{V}[\tilde{S}]}{\epsilon^2}=\frac{\|h\|_1^2}{t\epsilon^2}\left(\frac{4}{M}+5\right).
\end{align}
For a constant failure probability, it sufficies to take $M\in\mathcal{O}(1)$ and $t\in\mathcal{O}(\frac{1}{\epsilon^2}\|h\|_1^2)$.
The dependence on the failure probability can be exponentially improved via median boosting \cite{10.5555/646239.683352}.
\end{proof}

\begin{figure}[htbp]
\begin{algorithm}[H]
\caption{$\ell_1$-sampling operator shadow\label{algorithmL1}}
\renewcommand{\algorithmicrequire}{\textbf{Input:}}
\renewcommand{\algorithmicensure}{\textbf{Output:}}
\begin{algorithmic}[1]
\Require{Query and $\ell_1$-sampling access for the Pauli coefficients of an observable $O$, copies of an unknown quantum state $\rho$}
\For{$t=1$ to $T$}
\State Sample an index $j_t\in [L]$ with the $\ell_1$-sampling access
\State Implement the corresponding Pauli operator $P_{j_t}$ to the state $\rho$ on quantum device and measure in the computational basis with constant (e.g., $4$) times.
\State Compute the mean of the measurement outcomes as in Lemma \ref{lemmaPauliMeasurement} and store it as $Z_{j_t}$.
\State Compute $\tilde{X}_t:=\mathrm{sign}(h_{j_t})\left(2Z_{j_t}-1\right)\|h\|_1$ as in Theorem \ref{theoremL1}
\EndFor
\Ensure{$\tilde{S}=\frac{1}{T} \sum_{t=1}^T \tilde{X}_t$}
\end{algorithmic}
\end{algorithm}
\end{figure}

In the next subsection, we describe how to construct the sampling access explicitly.
We also describe the $\ell_2$-sampling operator shadow in Appendix ~\ref{AppendixL2}, whose sample complexity is $\mathcal{O}(L\|h\|_2^2/\epsilon^2)$.
$\ell_1$-sampling operator shadow is better than $\ell_2$-sampling operator shadow on the sample complexity, which can be seen by using the standard norm identity as Eq.~(\ref{eqL1}).
Note that here, $h$ is a $L$-dimension vector.

\subsection{Explicit construction of sampling access} \label{sec.construct}
In many practical problem settings like fidelity estimation \cite{PhysRevLett.106.230501}, described in Section \ref{section.fidelity}, the Pauli decompositions of observables are already given. 
However, if we hope to sample according to these Pauli coefficients, we also need to construct a data structure that makes sampling access possible.
Here, we describe the explicit construction of sampling access used in Theorem \ref{theoremL1}.

We start with the step to compute the Pauli coefficients. 
If the Pauli coefficients are already known, one can skip to the next step.
For a $2^n$-dimensional observable $O=\sum_{i=1}^{2^{2n}}h_i P_i$, one can compute each Pauli coefficient $h_i$ by $h_i=\mathrm{Tr}[P_i O]/2^n$. 
The heaviest part of computing is the matrix multiplication.
Note that an $n$-qubit Pauli string has only one nonzero element at each row and column, and has $2^n$ nonzero entries in total.
Based on this observation, one can use time $\mathcal{O}(2^{3n})$ to compute all Pauli coefficients \cite{romero2023paulicomposer}. 
If the observable $O$ is $k$-local, i.e., only nontrivially acts on $k$ qubits, the time complexity is $\mathcal{O}(2^{3k})$.

Once we have the values of each $h_k$, we can construct a binary search tree structure as Fig.~\ref{fig:enter-label}, see \cite{Walker1974,Vose1991,Kerenidis2017}. 
It takes time $\widetilde{\mathcal{O}}(2^{2n})$ to construct the data structure in the worst case, where the tilde notation omits $\log$ factors.
After constructing the data structure, sampling access takes $\mathcal{O}(n)$ time per sample.
Note that updating the data structure is efficient, i.e., it takes $\mathcal{O}({\rm poly} (n))$ time to read and update an entry.
For $k$-local observable, it takes $\widetilde{\mathcal{O}}(2^{2k})$ for the data structure construction and $\mathcal{O}(k)$ for each sampling access.
We summarize the results of constructing the data structure and implementing one sampling access under different conditions in Table.~\ref{table.datastructure}.

\begin{table}[H]
\begin{center}
\begin{tabular}{ |p{6cm}|p{3cm}|p{3cm}|  }
\hline
Setting & Construction & Sampling access\\
\hline \hline
General   & $\widetilde{\mathcal{O}}(2^{3n})$ & $\mathcal{O}(n)$\\
\hline
$k$-local &   $\widetilde{\mathcal{O}}(2^{3k})$ & $\mathcal{O}(k)$  \\
\hline
Pauli coefficients known &$\widetilde{\mathcal{O}}(2^{2n})$ & $\mathcal{O}(n)$\\
\hline
$k$-local \& Pauli coefficients known    &$\widetilde{\mathcal{O}}(2^{2k})$ & $\mathcal{O}(k)$ \\
\hline
\end{tabular}
\end{center}
\caption{Time complexity for explicit construction of the $\ell_1$-sampling access in different settings. }
\label{table.datastructure}
\end{table}

\begin{figure}
\centering
\begin{tikzpicture}[level distance=1.5cm,
level 1/.style={sibling distance=3cm},
level 2/.style={sibling distance=1.5cm},
level 3/.style ={level distance=1cm}]
\node {$\|h\|_1$}
child {node {$|h_1|+|h_2|$}
child {node {$|h_1|$} child{node {$\mathrm{sign}(h_1)$}}}
child {node {$|h_2|$}child{node {$\mathrm{sign}(h_2)$}}}
}
child {node {$|h_3|+|h_4|$}
child {node {$|h_3|$}child{node {$\mathrm{sign}(h_3)$}}}
child {node {$|h_4|$}child{node {$\mathrm{sign}(h_4)$}}}
};
\end{tikzpicture}
\caption{Data structure for $\ell_1$-sampling importance shadow of an operator with four Pauli coefficients, i.e., given the 4-dimensional vector $(h_1, h_2, h_3, h_4)^T$.}
\label{fig:enter-label}
\end{figure}

\subsection{Generalization to multiple observables}

Here, we discuss how to generalize to multiple observables. 
Mainly, we consider two problems: (1) estimate the expectation values of multiple observables, (2) estimate the expectation value of a linear combination of observables.

\begin{prob}[Estimation for multiple observables] \label{task.Multi}
Assume we are given $M$ observables $O_i$, $i\in [M]$, and many copies of an unknown quantum state $\rho$.
With constant success probability, estimate each $\mathrm{Tr}[O_i\rho]$ with additive precision $\epsilon$.
\end{prob}

\begin{prob}[Estimation for linear combination of observables] \label{task.Combine}
Assume we are given $M$ observables $O_i$, $i\in [M]$, real coefficients $w_i$, and many copies of an unknown quantum state $\rho$.
With constant success probability, estimate $\mathrm{Tr}[A\rho]:=\mathrm{Tr}[\sum_i w_i O_i\rho]$ with additive precision $\epsilon$.
\end{prob}

Problem \ref{task.Multi} is exactly the same question
considered in classical shadow \cite{huang_predicting_2020}.
Problem \ref{task.Combine} corresponds to cases like estimating the energy of $k$-local Hamiltonians, where each term of the Hamiltonian acts nontrivially on at most $k$ qubits.
\begin{corollary}[$\ell_1$-sampling operator shadow for multiple observables]\label{coro.multi}
There is an algorithm that solves Problem~\ref{task.Multi} with $\mathcal{O}(M\log M\|h\|_1^2/\epsilon^2)$ samples.
\end{corollary}
\begin{proof}
Simply apply Theorem \ref{theoremL1} for each observable and independently repeat $\mathcal{O}(\log M)$ times to achieve the constant success probability by the union bound.
\end{proof}

\begin{corollary}[Estimate linear combination of observables]\label{coro.lin}
There is an algorithm that solves Problem~\ref{task.Combine} with $\mathcal{O}(\max_i \|h^i\|_1^2 \|w\|_1^2/\epsilon^2)$ samples and $\widetilde{\mathcal{O}}(M 2^{3n})$ running time, where $\|w\|_1=\sum_i |w_i|$ and $\|h^{i}\|$ is the $\ell_1$ norm of the Pauli coefficients of the $i$-th observable.
\end{corollary}
\begin{proof}
First, we sum up the Pauli coefficients for each $O_i$, i.e., $A=\sum_i \sum_j w_i h_j^i P_j^i$, which has at most $4^n$ terms.
If the Pauli coefficients are not given, we can compute them as in Section~\ref{sec.construct}.
We define $h\in \mathbb R^{L}$ such that $A=\sum_{k=1}^{L} h_k P_k$ and construct the sampling data structure as in Section~\ref{sec.construct}, where $L=4^n$.
These preparations will take time $\widetilde{\mathcal{O}}(M2^{3n})$ if Pauli coefficients are not given and $\widetilde{\mathcal{O}}(M2^{2n})$ if they are given.
We use Theorem \ref{theoremL1} to estimate $\mathrm{Tr}[A\rho]$ directly.
By the direct calculation, we have 
\begin{align}
\|h\|_1 \leq \sum_i \sum_j |w_i h_j^i|\leq \max_i \|h^i\|_1 \sum_i|w_i|=\max_i \|h^i\|_1 \|w\|_1.
\end{align}
Therefore, the sample complexity is $\mathcal{O}(\max_i \|h^i\|_1^2 \|w\|_1^2/\epsilon^2)$.
\end{proof}

\section{Comparison with classical shadow}

In this section, we first describe the previous work about classical shadow \cite{huang_predicting_2020}, and discuss the comparison.
The essence of the classical shadow is to construct a classical sketch of a density matrix $\rho$ via random measurements and classical processing.
The given quantum state is evolved with a unitary chosen randomly from a fixed ensemble.
Widely considered ensembles are Pauli operators and Clifford unitaries.
Here, we focus on the Pauli operator ensembles and $k$-local observables, in which case we can make a direct comparison.

Given an unknown quantum state $\rho$, choose uniformly a unitary $U$ from Pauli ensemble and apply to the quantum state, i.e., $\rho \rightarrow U\rho U^\dagger$.
Let a computational basis measurement outcome be  $\hat{b}\in \{0,1\}^n$.
The averaged mapping from $\rho$ to the state implied by the measurement outcomes can be viewed as a quantum channel:
\begin{align}
\mathcal{M}(\rho)= \mathbbm E\left[ U^\dagger |\hat{b}\rangle\langle \hat{b}| U \right].
\end{align}
Inverting the channel formally gives the quantum state $\rho$, i.e., $ \rho = \mathbbm E\left[\mathcal M^{-1} \left(U^\dagger |\hat{b}\rangle\langle \hat{b}| U \right)\right]$. Given a set of actual measurement outcomes,  define the classical snapshot $\hat \rho(\hat b) := \mathcal M^{-1}\left( U^\dagger |\hat{b}\rangle\langle \hat{b}| U \right)$. By definition $\rho(\hat b)$ is an unbiased estimator of $\rho$, i.e., $\mathbb{E}[\hat{\rho}(\hat{b})]=\rho$.
For the case of random Pauli measurements, the mapping has a closed form
\begin{align}\label{eq.shadow}
\hat{\rho}(\hat b)=\bigotimes_{j=1}^n \left(3U_j^\dagger |\hat{b}\rangle\langle \hat{b}| U_j-\mathbb{I}\right),
\end{align}
where $\{U_j\}$ is the sequence of the random unitaries, i.e., $U=U_1\otimes U_2\cdots\otimes U_n$.
With this, one can estimate the expectation value of observables with an unknown quantum state.
The theoretical guarantee is given as follows.

\begin{lemma}[Classical shadow with Pauli measurements \cite{huang_predicting_2020}] \label{lemmaShadow}
Let $\epsilon\in(0,1]$. Given multiple copies of an $n$-qubit quantum state $\rho$, fix a measurement primitive $\mathcal U$ of randomly chosen Pauli measurements and a collection $O_{1}, \ldots, O_{M}$ of $2^{n} \times 2^{n}$ $k$-local observables.
To estimate $\mathrm{Tr}[O_i\rho]$ for each $i\in [M]$ up to precision $\epsilon$ with constant probability, using the above procedure with $\mathcal U$ it sufficies to take 
$$\mathcal{O}\left(\max_{i\in [M]}\|O_i\|_{\infty}^2 4^k\log M/\epsilon^2\right)$$ 
independent samples, where $\|\cdot\|_\infty$ is the spectral norm.
\end{lemma}

Here, we add the analysis of the time complexity for the classical shadow method.
The time complexity arises from two parts: First, construct the shadow $\hat{\rho}(\hat{b})$ for each measurement outcome $\hat{b}$. Second, compute $\mathrm{Tr}[O\hat{\rho}(\hat{b})]$ and take the average.
For the first step, it takes time $\mathcal{O}(n\cdot T_{\mathrm{samp}})$, where $T_{\mathrm{samp}}$ is the sample complexity, and $n$ arising from the need to compute the $\mathbb{C}^{2\times 2}$ matrix representation for each qubit, as Eq.~(\ref{eq.shadow}).
For the second step, 
note that the shadow $\hat{\rho}(\hat b)$ is of the form $\hat{\rho}(\hat b)=\hat{\rho}_1(\hat{b})\otimes\cdots\otimes\hat{\rho}_n(\hat{b})$.
One can write a $k$-local observable $O$ in the form $O=\hat{O}\otimes \mathbb{I}_{n-k}$, where $\hat{O}$ acts nontrivially on at most $k$ qubits.
We have $\mathrm{Tr}[O \hat\rho(\hat b)]=\mathrm{Tr}[\hat{O}\hat{\rho}_1(\hat{b})\otimes\cdots\otimes \hat{\rho}_k(\hat{b})]$.
Computing this expectation value will take time $\mathcal{O}(2^{2k})$, which can be seen from the trace computation among two $\mathbb{C}^{2^k\times 2^k}$ matrices.
Since we do the computation for every snapshot, in total, the time complexity of the classical shadow is $\mathcal{O}((n+2^{2k})\cdot T_{\mathrm{samp}})$. We leave open the possibility of more efficient post-processing.

\subsection{Multiple local observables}

Here we consider Problem~\ref{task.Multi}, and further assume all observables are $k$-local.
Recall that in this work, a $k$-local observable means that the observable acts nontrivially on at most $k$-qubits.
As shown in Corollary~\ref{coro.multi}, the $\ell_1$-sampling operator shadow scales with $\mathcal{O}(M)$, which means that for this factor, Theorem \ref{theoremL1} is exponentially worse than the classical shadow. 
However, the dependency of Theorem \ref{theoremL1} on other parameters is actually better than the classical shadow.
Hence, we can derive a range of $M$ in which the sample complexity of Theorem \ref{theoremL1} is better than Lemma \ref{lemmaShadow}.

\begin{lemma}
Let us be given a collection $O_{1}, \ldots, O_{M}$ of $2^{n} \times 2^{n}$ $k$-local observables, where $h^j$ is the coefficient vector for Pauli decomposition of $O_j$.
For $M/\log(2M) \leq C 4^k \|O_{i^\ast}\|^2_{\infty}/\|h^{j^\ast}\|_1^2$, where 
$i^\ast :=\arg \max_{i\in [M]} \|O_i\|_{\infty}$, $j^\ast = \arg \max_{j\in [M]} \|h^{j}\|_1$,  and $C>0$ a constant, the sample complexity of Theorem \ref{theoremL1} is better than Lemma \ref{lemmaShadow}.
\end{lemma}
\begin{proof}
Using the importance sampling, we obtain an upper bound on the total number of samples of 
$C_1 M \frac{\|h^{j^\ast}\|_1^2}{\epsilon^2}$,
where $j^\ast = \arg \max_{j\in [M]} \|h^{j}\|_1$ and $C_1>0$ is a constant.
On the other hand, for the shadows we have the result from Lemma \ref{lemmaShadow} being $C_2 (4^k \|O_{i^\ast}\|^2_{\infty}) \log(2M)/\epsilon^2$, where
$i^\ast :=\arg \max_{i\in [M]} \|O_i\|_{\infty}$ and $C_2>0$ is a constant.
Hence, if 
$$C_1 M\|h^{j^\ast}\|^2_1 \leq C_2 4^k \|O_{i^\ast}\|^2_{\infty}\log(2M),$$
we obtain a better sample complexity for the importance sampling strategy.
To finish the proof, we only need to show that $4^k \|O_{i^\ast}\|^2_{\infty}/\|h^{j^\ast}\|_1^2 \geq 1$. Since these observables are $k$-local, let $O_i=\widetilde{O}_i\otimes \mathbb{I}$, where $\widetilde{O}_i$ acts nontrivially on $k$-qubits.
We have $O_i=\sum_{j=1}^{L}h_j^i P_j\otimes \mathbb{I}$, with $L=4^k$.
Note that from Eq.~(\ref{eqPauli}) we obtain
\begin{align}
\mathrm{Tr}[O^2]&= \mathrm{Tr}\biggl[\sum_{i,j}h_ih_j P_iP_j\biggl]=\mathrm{Tr}\biggl[\sum_i h_i^2\cdot \mathbb{I}\biggl]=2^n \|h\|_2^2. \label{eqL2}
\end{align}
Combining Eq.~(\ref{eqL2}) and Eq.~(\ref{eqL1}), we have 
\be
\|h^{j^\ast}\|_1\leq \sqrt{L}\|h^{j^\ast}\|_2=&2^{\frac{k}{2}} \left(\mathrm{Tr}[\widetilde{O}_{j^\ast}^2]\right)^{\frac{1}{2}}
= 2^{\frac{k}{2}} \|\widetilde{O}_{j^\ast}\|_F
\leq  2^k \|\widetilde{O}_{j^\ast}\|_{\infty} \leq  2^k \|\widetilde{O}_{i^\ast}\|_{\infty},
\ee
where the second equality comes from the definition of Frobenius norm, and the second-to-last inequality comes from Lemma \ref{lemmaNorm}.
\end{proof}

In many practical cases, we may only need to predict constant or small expectation values of observables, and in such cases, Theorem \ref{theoremL1} may have better performance.
We summarize the result in Table.~\ref{table.multi}.
We denote sample complexity by $T_{\rm samp}$.

\begin{table}[H]
\begin{center}
\begin{tabular}{ |p{4cm}|p{6cm}|p{4cm}|  }
\hline
Methods & Sample complexity $T_{\rm samp}$ & Time complexity\\
\hline \hline
Classical shadow \cite{huang_predicting_2020}   & $\mathcal{O}\left(\max_{i\in [M]}\|O_i\|_{\infty}^2 4^k\log M/\epsilon^2\right)$ & $\mathcal{O}( (n+M 2^{2k})\cdot T_{\rm samp})$\\
\hline
This work &   $\mathcal{O}(M\|h\|_1^2/\epsilon^2)$ & $\mathcal{O}(k\cdot T_{\rm samp} + M2^{3k}) $  \\
\hline
\end{tabular}
\end{center}
\caption{Comparison for multiple local observables}
\label{table.multi}
\end{table}

\subsection{Linear combination of local observables}
Here we consider Problem~\ref{task.Combine}, and further assume all observables are $k$-local.
For $M$ $k$-local observables $O_i$, in the worst case $A= \sum_i w_i O_i$ will be $Mk$-local.
If we directly use Lemma \ref{lemmaShadow} to estimate $\mathrm{Tr}[A\rho]$, the sample complexity scales exponentially in terms of $M$.
To achieve the task, with classical shadow we have 
\begin{align}
\left|\sum_i w_i \mathrm{Tr}[O_i\rho]- \sum_i w_i \widetilde{\mathrm{Tr}}[O_i\rho]\right|
\leq \sum_i |w_i| \left|\mathrm{Tr}[O_i\rho]- \widetilde{\mathrm{Tr}}[O_i\rho]\right|\leq \epsilon,
\end{align}
where $\widetilde{\mathrm{Tr}}[\cdot]$ denotes the estimated value.
It suffices to estimate each expectation value with precision $\epsilon/\|w\|_1$, where $\|w\|_1=\sum_i |w_i|$.
Since there are at most $M$ terms, and computing each term $\mathrm{Tr}[O_i\rho]$ takes time $\mathcal O(2^{2k})$, in total it takes time $\mathcal{O}(M2^{2k})$.
As mentioned, for each sample, it takes $\mathcal{O}(k)$ time for the post-processing.
Based on these, we have the following result.
\begin{corollary}
Lemma~\ref{lemmaShadow} can achieve the task described in Problem~\ref{task.Combine} with $$T_{\rm samp} \in \mathcal{O}\left(\frac{1}{\epsilon^2}\log M 4^k \|w\|_1^2 \max_i \|O_i\|_{\infty}^2 \right)$$
copies of the quantum state and $\mathcal{O}((n+M 2^{2k})\cdot T_{\rm samp})$ running time.
\end{corollary}

The sample complexity of Theorem \ref{theoremL1}, which is described in Lemma~\ref{coro.lin}, is strictly better than Lemma \ref{lemmaShadow}.
This can be seen by noticing that $\max_i \|h^i\|_1\leq 2^k \max_i \|O_i\|_{\infty}$.
However, if we need to estimate the Pauli coefficients,
the time complexity of Theorem~\ref{theoremL1} is worse than the classical shadow.
We summarize the result in Table.~\ref{table.linear}.

\begin{table}[H]
\begin{center}
\begin{tabular}{ |p{3cm}|p{7.5cm}|p{4cm}|  }
\hline
Methods & Sample complexity $T_{\rm samp}$ & Time complexity\\
\hline \hline
Classical shadow  \cite{huang_predicting_2020} & $\mathcal{O}\left(\log M 4^k \|w\|_1^2 \max_i \|O_i\|_{\infty}^2 /\epsilon^2\right)$ & $\mathcal{O}((n+M 2^{2k})\cdot T_{\rm samp})$\\
\hline
This work &  $\mathcal{O}(\max_i \|h^i\|_1^2 \|w\|_1^2/\epsilon^2 )$  & $\mathcal{O}(k\cdot T_{\rm samp} + M2^{3k})$  \\
\hline
\end{tabular}
\end{center}
\caption{Comparison for linear combination of local observables}
\label{table.linear}
\end{table}

\section{Application for fidelity estimation}\label{section.fidelity}

In this section, we consider an application called the \textit{fidelity estimation}.

\begin{prob}[Fidelity estimation] \label{task.fidelity}
For two $n$-qubit quantum states $\rho$ and $\sigma$, estimate fidelity $$F(\rho,\sigma)=\mathrm{Tr}\left[\sqrt{\sqrt{\rho}\sigma \sqrt{\rho}}\right],$$
to a given additive error with high probability.
\end{prob}

We consider the same setting as in Ref.~\cite{PhysRevLett.106.230501}:
let $\rho$ be the target pure state we want to prepare, and $\sigma$ be the exact quantum state we prepare on the quantum device.
They also implicitly assume that we have sampling access to the Pauli coefficients of $\rho$.
Recall that if one state is a pure state, the formula can be simplified as $F(\rho,\sigma)=\mathrm{Tr}[\rho \sigma]$.
For this case, fidelity estimation has been proved to be BQP-complete in Ref.~\cite{PhysRevA.108.012409}.

Flammia and Liu considered using the importance sampling method for this task \cite{PhysRevLett.106.230501}, for which this work can be understood as an improved version along the same direction.
Recall that for a pure state $\rho=\sum_i h_iP_i$, $\mathrm{Tr}[\rho^2]=1=2^n \|h\|_2^2$.
The probability of index $i$ with $\ell_2$-sampling is directly $(\mathrm{Tr}[ P_i\rho])^2/2^n$,
which is the key observation in Ref. \cite{PhysRevLett.106.230501}.
We have 
\begin{align}
\mathrm{Tr}[\rho\sigma]=\frac{1}{2^n}\sum_i^L  \mathrm{Tr}[P_i\rho]\mathrm{Tr}[P_i\sigma].
\end{align}
By estimating the expectation value of the random variable $$X_i=\mathrm{Tr}[P_i\sigma]/\mathrm{Tr}[P_i\rho ],$$ where index $i$ is sampled according to $(\mathrm{Tr}[ P_i\rho])^2/2^n$,
one can estimate the $\mathrm{Tr}[\rho\sigma]$ with error up to $\epsilon>0$.
The sample complexity is $\mathcal{O}\left(2^n/\epsilon^2\right)$.
If we do not have sampling access, we can construct it as in Section \ref{sec.construct}.

This result matches our result of the $\ell_2$-sampling, i.e., $\mathcal{O}(L\|h\|_2^2/\epsilon^2)$, where the equality can be seen from Eq.~(\ref{eqL2}).
By using the $\ell_1$-sampling operator introduced in this work, the sample complexity is 
$\mathcal{O}\left(\|h\|_1^2/\epsilon^2\right).$
As mentioned, $\ell_1$-sampling is better than $\ell_2$-sampling.
There is a specific case that will make $\ell_1$-sampling and $\ell_2$-sampling equal, which is all $|h_i|$ have the same value. In this case, the random variable of $\ell_1$-sampling and $\ell_2$-sampling will be identical, resulting the same performance of these two methods. An example of such a case is the $n$-qubit GHZ state $|\psi\rangle\langle\psi|=\frac{1}{2}(|0\rangle\langle0|^{\otimes n}+|0\rangle\langle1|^{\otimes n}+|1\rangle\langle0|^{\otimes n}+|1\rangle\langle1|^{\otimes n})$, with all coefficients $|h_i|$ being $1/2^{n-1}$.

\section{Numerical simulation}

In our first numerical experiment, we evaluate the performance of the operator shadow in fidelity estimation tasks. 
Similar to \cite{PhysRevLett.106.230501}, we consider a system size of $8$ qubits and a Haar-random pure state generated by Haar-random unitaries is selected as the target state. 
To this state, a 10\% depolarizing error is introduced, resulting in an expected final fidelity of approximately 0.9. 
We compare the performance of $\ell_1$- and $\ell_2$-sampling, with the latter corresponding to the method used in~\cite{PhysRevLett.106.230501}. 
Let $X$ be the estimated fidelity. The metric in this experiment is the absolute error $|X-\mathbbm E(X)|$ and the failure probability $\mathbb{P}(|X-\mathbbm E(X)|\geq \epsilon)$, and $\epsilon$ is set to 0.03. 
With the number of samples $M$ required to estimate each Pauli coefficient held constant, prior deductions suggest that the sample complexity for $\ell_1$-sampling is $\mathcal{O}\left(\|h\|_1^2/\epsilon^2\right)$, which is notably better than $\ell_2$-sampling. 
The results, as shown in Fig.~\ref{fig:fidelity_estimation}, indicate that across all sample size ranges, $\ell_1$-sampling  demonstrates lower absolute errors and failure probabilities compared to $\ell_2$-sampling.

\begin{figure}[htb]
\centering
\includegraphics[width=0.9\linewidth]{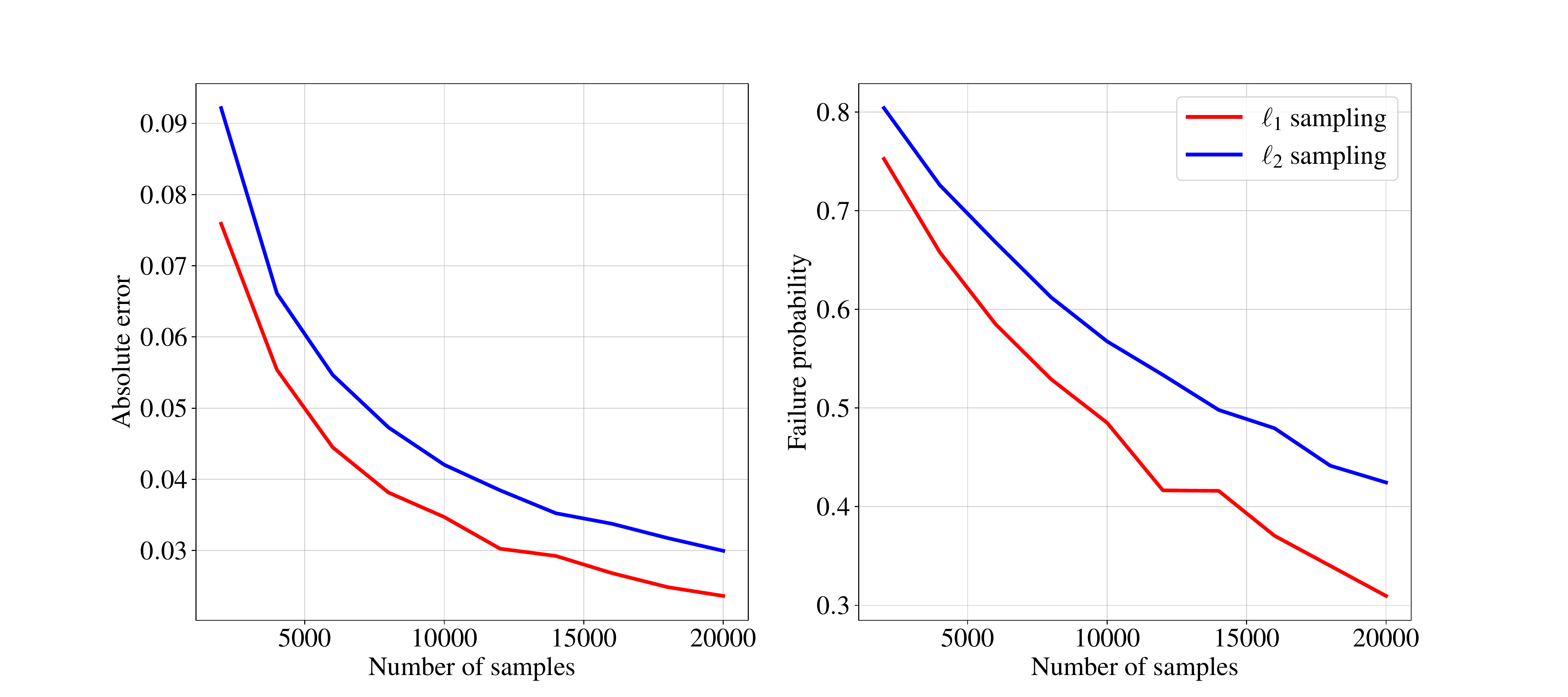}
\caption{Numerical results of the fidelity estimation task}
\label{fig:fidelity_estimation}
\end{figure}

In the second numerical experiment, we compare the performance of classical shadow and operator shadow in estimating the linear combination of observables, as outlined in Problem~\ref{task.Combine}. 
The target system is a rotated surface code~\cite{bombin2007optimal, anderson2013homological}, defined by the Hamiltonian:
\begin{equation}
    H = -\sum_{s=1}^{m_s} S_s - \sum_{p=1}^{m_p} P_p \;,
    \label{eq:rsur}
\end{equation}
where
\begin{equation}
    S_s = \prod_{i\in q_s} Z_i \quad P_p = \prod_{j\in q_p} X_j \;.
    \label{eq:stabilizer}
\end{equation}

The $S_s$ and $P_p$ terms represent Z and X stabilizer generators. All $m=m_s+m_p$ stabilizer generators form the stabilizer group to protect the code space. Each generator is a product of Pauli $Z$ or $X$ operators acting on a set of qubits $q_s$ or $q_p$. Since $S_s$ and $P_p$ are stabilizer generators, they will be commute with each other, ensuring that the ground state of such a Hamiltonian is the $+1$ eigenstates of all generators. A pictorial demonstration of a rotated surface code with $n=9$ qubits and $m=8$ stabilizer generators is shown in Fig.~\ref{fig:rsur}.

\begin{figure}[htb]
\centering
\includegraphics[width=0.4\linewidth]{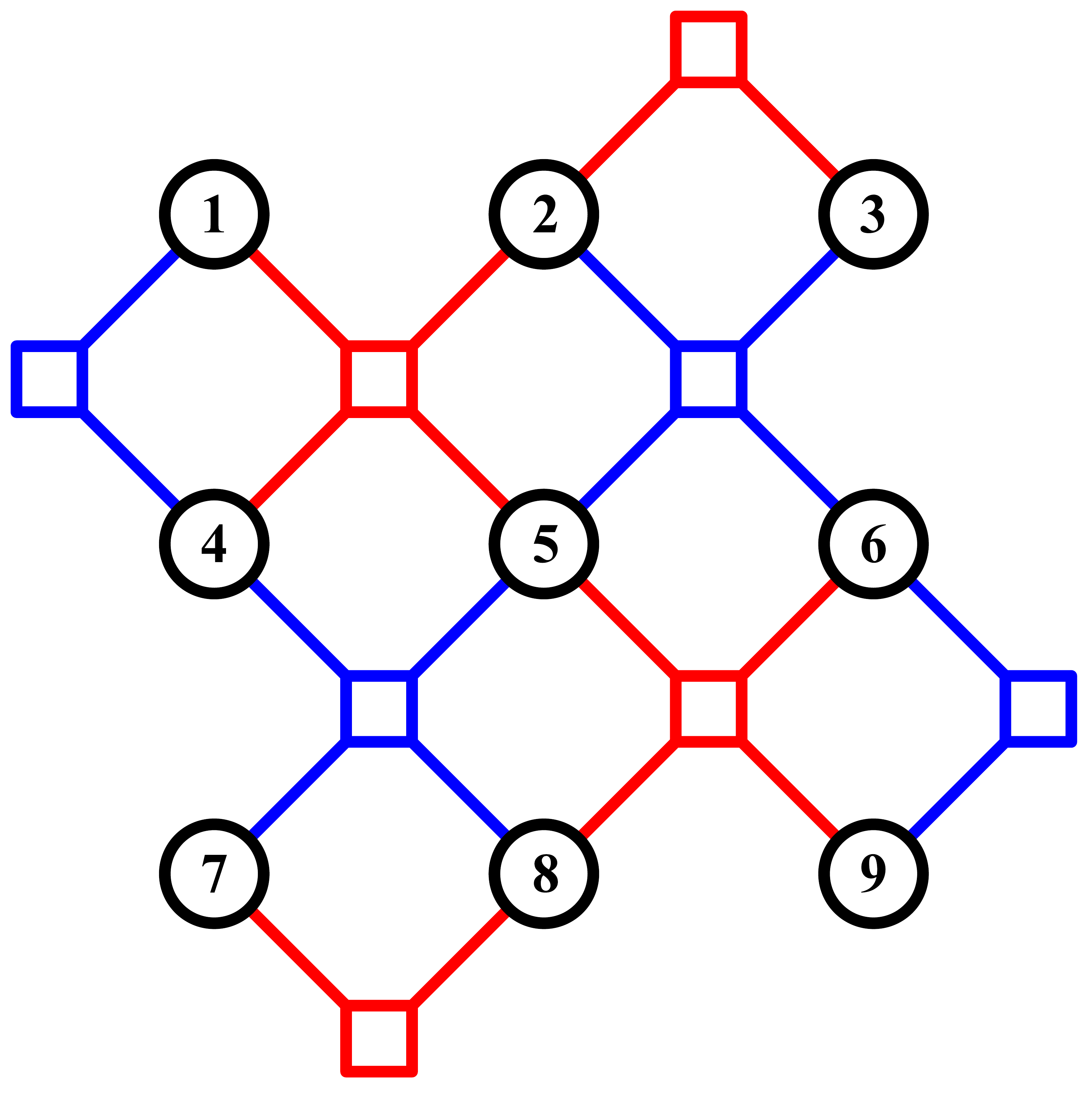}
\caption{Rotated surface code with $n=9$. Here the red squares and blue squares are representing X and Z stabilizer generators respectively. The $S$ terms are $X_2X_3$, $X_1X_2X_4X_5$, $X_5X_6X_8X_9$ and $X_7X_8$, and the $P$ terms are $Z_1Z_4$, $Z_2Z_3Z_5Z_6$, $Z_4Z_5Z_7Z_8$ and $Z_6Z_9$.  }
\label{fig:rsur}
\end{figure}

When the Hamiltonian serves as the observable of interest, it is recognized as a $4$-local global observable, as the maximum weight of stabilizer gnerators is 4 in a rotated surface code, with $m=n-1$ local observables.  
We will add coefficients which obeys the normal distribution to each $S_s$ and $P_p$ term to make the $w_i$ in the linear combination of operators nontrivial.
Given that all $O_i$s are Pauli operators, both their spectral norms and the $\ell_1$ norms of their Pauli coefficients will be $1$. 
These lead to the sample complexities $\mathcal{O}\left(\log(n-1) 4^4 \|w\|_1^2 /\epsilon^2\right)$ for classical shadow and $\mathcal{O}\left(\|w\|_1^2 /\epsilon^2\right)$ for operator shadow. 
In this scenario, as well as in cases where the desired observables are $k$-local Hamiltonian composed of Pauli operators,  the benefits of operator shadow over classical shadow is evident.

In the numerical simulations, we choose a rotated surface code with system size $n=9$. 
To increase the difficulty of such simulation, we set the system is considered in the state
\begin{equation}
    \rho = 0.9\rho_0 + 0.1 \rho_{\text{Haar}} \;,
\end{equation}
where $\rho_0$ is the ground state of the rotated surface code and $\rho_{\text{Haar}}$ is a Haar-random state. 
In such state, the actual value of the Hamiltonian estimation will be around $0.9E_0$ where $E_0$ is the ground state energy.
The corresponding numerical results are presented in Fig~\ref{fig:energy_estimation}. 
It is observed that operator shadow outperforms classical shadow in both absolute error and failure probability. 
It is important to note that while there is a derandomized version of classical shadow~\cite{huang2021efficient}, it is not included in this comparison due to its deterministic measurement scheme.

\begin{figure}[htb]
\centering
\includegraphics[width=0.9\linewidth]{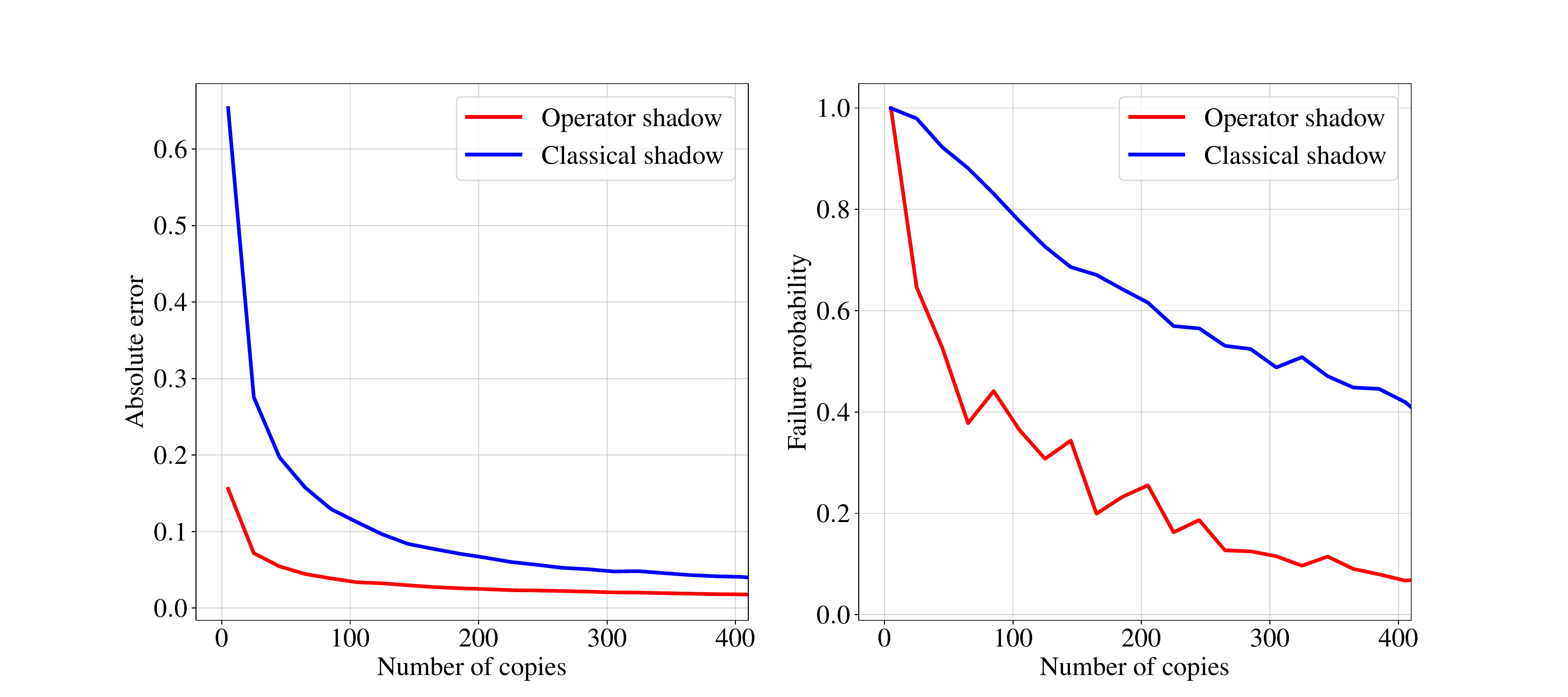}
\caption{Numerical results of the Hamiltonian estimation task. Here the $\epsilon$ to estimate the failure probability is set to be 0.04.}
\label{fig:energy_estimation}
\end{figure}

\section{Conclusion}

We provide a simple method for estimating the expectation value of observables with an unknown quantum state with a comprehensive analysis including how to construct the data structure and both the sample and time complexity. 
The key idea is to do the importance sampling with the Pauli decompositions of the observables. Given this sampling, the method is rather easy to implement experimentally with very little post-processing of the measurement data.
For estimating expectation values of many local observables, which are considered in classical shadow, our method is better only when the number of observables is small.
To estimate the expectation value for the linear combination of local observables, our method is better for all parameters in terms of sample complexity.
Regarding time complexity, our method has to precompute the  sampling data structure in time about $O(2^{3k})$. For classical shadows, we find a post-processing time complexity  of $O(2^{4k})$, the product of  sample complexity and matrix multiplication time. We leave open ways to improve the time complexity of the classical shadow method.
However, if we already know the Pauli coefficients and the sampling data structure, the time complexity is dominated by the sample complexity and hence better than the classical shadow.
This work addresses the fundamental problem of measurement and learning of quantum systems via a method for sampling the operator to be measured and could show benefits in rapidly discovering properties of actual quantum systems.
\section{Acknowledgement}
We thank Bin Cheng for helpful discussions. This research is supported by the National Research Foundation, Singapore, and A*STAR under its CQT Bridging Grant and grant NRF2020-NRF-ISF004-3528.
The authors thank the anonymous reviewers of AQIS2023 for their advice to improve this work.

\bibliographystyle{unsrt}
\bibliography{ref}

\appendix

\newpage
\section{Proof of Lemma \ref{lemmaNorm}}\label{appen.norm}

By definition, we have
$$
\|A\|_F=\sqrt{\sum_i\sum_j |a_{ij}|^2}
= \sqrt{\sum_i \sigma_i(A)^2}
\leq \sqrt{N\max_i \sigma_i(A)^2}=\sqrt{N} \|A\|_{\infty}.$$

\section{$\ell_2$-sampling operator shadow}\label{AppendixL2}
\subsection{Perfect $\ell_2$-sampling}
\begin{defn}[$\ell_2$-sampling access for $O$]\label{defSampling}
Define the sampling access such that we are able to sample an index $k$ with probability 
\be
Q_k :=\frac{ \vert h_k\vert ^2 }{\Vert h \Vert_2 ^2 }.
\ee
\end{defn}

\begin{lemma} [Perfect $\ell_2$-sampling]\label{lemmaDenseQuery}
Assume sampling access as in Definition \ref{defSampling}. Also assume query access to $h_k$ and  $\mathrm{Tr}[P_k\rho]$ for all $k$, i.e., given an index $k$ we obtain the respective element $h_k$ or $\mathrm{Tr}[P_k\rho]$ to very high accuracy (numerical precision). In addition assume knowledge of the norms $ \Vert h \Vert^2$ and $\Vert  \psi \Vert_P^2$, where $\Vert  \psi \Vert_P^2 := \sum_{k=1}^{L} \vert \mathrm{Tr}[P_k\rho]\vert^2.$
Let $\epsilon>0$.
Then we can estimate 
$\mathrm{Tr}[O\rho]$ to additive accuracy $  \Vert  \psi \Vert_P  \Vert h \Vert \epsilon$ in $\Ord{1/\epsilon^2}$ samples and queries with constant success probability. 

\end{lemma}

\begin{proof}
The algorithm is as follows. Sample according to $Q_k$  an index $k$ and query the respective elements to form the random variable 
\be Y:=\frac{\mathrm{Tr}[P_k\rho]\Vert h\Vert^2 } {h_{k} }.
\ee 
The expectation value is
\be
\mathbb{E}[Y] = \sum_{k=1}^{L} Q_k \frac{\mathrm{Tr}[P_k\rho]   \Vert h\Vert^2 } {h_{k} } 
= \sum_{k=1}^{L} \frac{ \vert h_k
\vert^2}{\Vert h \Vert^2 } \frac{\mathrm{Tr}[P_k\rho]   \Vert h\Vert^2 } {h_{k} } 
= \mathrm{Tr}[O\rho].
\ee
The variance is bounded as
\be
\mathbb{V}[ Y ] \leq \sum_{k=1}^{L}  Q_k \left( \frac{\mathrm{Tr}[P_k\rho]   \Vert h\Vert^2 } {h_{k} }   \right)^2 = \sum_{k=1}^{L}  \frac{ \vert h_k \vert^2}{\Vert h \Vert^2 }  \left( \frac{\mathrm{Tr}[P_k\rho]   \Vert h\Vert^2 } {h_{k} }   \right)^2 = \Vert  \psi \Vert_P^2  \Vert h \Vert^2.
\ee
Performing this $T$ times and taking the mean 
leads to the random variable 
\be
S= \frac{1}{T} \sum_{t=1}^T Y_t, 
\ee
with expectation value $ \mathrm{Tr}[O\rho]$ and variance $V[T] = V[Y]/t$.
Now estimate the failure probability to achieve an additive error $ \Vert  \psi \Vert_P  \Vert h \Vert \epsilon$ via Chebyshev's inequality to be
\be
\mathbb{P}\left [ \left \vert T - \mathbb{E}[T] \right \vert \geq  \Vert  \psi \Vert_P  \Vert h \Vert \epsilon\right] \leq \frac{ \mathbb{V}[T] }{ \Vert  \psi \Vert_P^2  \Vert h \Vert^2 \epsilon^2} \leq \frac{1}{\epsilon^2 t}.
\ee
For a constant failure probability (say $0.01$), we obtain the claim $T = \Ord{100/\epsilon^2}$. 
\end{proof}

\subsection{$\ell_2$-sampling operator shadow}
\label{proofL2}

\begin{theorem}[$\ell_2$-sampling operator shadow] \label{theoremL2}
Assume sampling access as in Definition \ref{defSampling}. Let there be given multiple copies of a quantum state $\rho$. Let $\epsilon>0$.
Then we can estimate 
$\mathrm{Tr}[O\rho]$ to additive accuracy $\epsilon$ in $ \Ord{L\|h\|_2^2/\epsilon^2}$ measurements with constant success probability.
\end{theorem}

\begin{proof}
The algorithm is given in Algorithm.~\ref{algorithmL2}.
Again sample with probability $Q_k$ an index $k$ and form the random variable of the measurement outcomes for that $k$, i.e., 
\be
\tilde Y:=  \left ( 2Z_k -1 \right) \frac{ \Vert h \Vert_2^2} {h_{k} }
\ee  see  Lemma \ref{lemmaPauliMeasurement}. For fixed $k$, the expectation value is simply 
$\mathbb{E}_k[\tilde Y] = \mathrm{Tr}[P_k\rho]  \frac{ \Vert h \Vert_2^2} {h_{k} }$ and the variance is 
\be
\mathbb{V}_k[\tilde Y] \leq \frac{ 4 \Vert h \Vert_2^4} {h_{k}^2 } \mathbb{E}_k \left [Z_k^2\right] = \frac{ 4\Vert h\Vert_2^4 }{h_{k}^2}   \left(\frac{1+\mathrm{Tr}[P_i\rho]}{2M}+\left(1-\frac{1}{M}\right)\left(\frac{1+\mathrm{Tr}[P_i\rho]}{2}\right)^2\right).
\ee 
For stochastic $k$ we take another expectation value
\be
\mathbb{E}[\tilde Y] =  \sum_k Q_k \mathrm{Tr}[P_k\rho] \frac{ \Vert h \Vert_2^2} {h_{k} } = \mathrm{Tr}[O\rho],
\ee
and the total variance can be bounded as 
\be
\mathbb{V}[ \tilde Y  ] &\leq&  \Vert h \Vert_2^4  \left ( 4 \mathbb{E} \left [Z_k^2 \right ]+ \mathbb{E} \left [ \left (\frac {1}{h_{k} } \right )^2\right]  \right)\\
&=&   \Vert h \Vert_2^4 \left( 4 \sum_{k=1}^L Q_k \left(\frac{1+\mathrm{Tr}[P_i\rho]}{2M}+\left(1-\frac{1}{M}\right)\left(\frac{1+\mathrm{Tr}[P_i\rho]}{2}\right)^2\right) + \frac{L}{ \Vert h \Vert^2} \right)\\
&=& \frac{2 \Vert h \Vert^2}{M} \left(1+ \mathrm{Tr}[P_k\rho]\right) + (L+4) \Vert h \Vert^2 \\
&\leq&  \frac{4\Vert h \Vert^2}{M}  + (L+4)\Vert h \Vert^2
\ee 
Now estimate the failure probability to achieve an additive error $\epsilon$ for the mean
\be
\tilde S : = \frac{1}{T} \sum_{t=1}^T \tilde Y_t, 
\ee
via Chebyshev's inequality to be
\be
\mathbb{P}\left [ \left \vert \tilde T - \mathbb{E}[\tilde T] \right \vert \geq \epsilon\right] \leq \frac{ \mathbb{V}[\tilde T] }{\epsilon^2} =\frac{\|h\|^2}{t \epsilon^2}\left( \frac{4}{M}  + 4+L \right).
\ee
For a constant failure probability, we obtain the claim $M\geq 4$ and $T = \Ord{L\|h\|^2/\epsilon^2}$.
\end{proof}

\begin{figure}[htbp]
\begin{algorithm}[H]
\caption{$\ell_2$-sampling operator shadow\label{algorithmL2}}
\renewcommand{\algorithmicrequire}{\textbf{Input:}}
\renewcommand{\algorithmicensure}{\textbf{Output:}}
\begin{algorithmic}[1]
\Require{Query and $\ell_2$-sampling access for the Pauli coefficients of an observable $O$, copies of an unknown quantum state $\rho$}
\For{$t=1$ to $T$}
\State Sample an index $j_t\in [L]$ with the $\ell_2$-sampling access
\State Implement the corresponding Pauli operator $P_{j_t}$ to the state $\rho$ on the quantum device and measure in the computational basis with constant (e.g., $4$) times.
\State Compute the mean of the measurement outcomes as in Lemma \ref{lemmaPauliMeasurement} and store it as $Z_{j_t}$.
\State Compute $\tilde{Y}_t:=\left(2Z_{j_t}-1\right)\|h\|_2^2/h_{j_t}$ as in Theorem \ref{theoremL2}
\EndFor
\Ensure{$\tilde{S}=\frac{1}{T} \sum_{t=1}^T \tilde{Y}_i$}
\end{algorithmic}
\end{algorithm}
\caption{Algorithm of $\ell_2$-sampling operator shadow}
\end{figure}

\end{document}